\documentclass[11pt]{article}

\usepackage[margin=1in]{geometry}
\usepackage{graphicx}
\usepackage[pdftex]{hyperref}
\usepackage{amsmath, amsthm, amssymb}
\usepackage{subfigure}
\usepackage{comment}
\usepackage{url}
\usepackage{xcolor}

\newcommand{\R}{\mathbb{R}}

\newcommand{\E}{\mathbb{E}}
\newcommand{\F}{\mathbb{F}}

\newcommand{\ket}[1]{| #1 \rangle}

\newcommand{\ip}[2]{\langle #1|#2 \rangle}

\newcommand{\be}{\begin{equation}}
\newcommand{\ee}{\end{equation}}
\newcommand{\bea}{\begin{eqnarray}}
\newcommand{\eea}{\end{eqnarray}}
\newcommand{\bes}{\begin{equation*}}
\newcommand{\ees}{\end{equation*}}
\newcommand{\beas}{\begin{eqnarray*}}
\newcommand{\eeas}{\end{eqnarray*}}

\newtheorem{thm}{Theorem}
\newtheorem*{thm*}{Theorem}

\newtheorem{lem}[thm]{Lemma}
\newtheorem*{lem*}{Lemma}

\theoremstyle{definition}
\newtheorem{dfn}{Definition}

\newcommand{\boxdfn}[2]{
\begin{figure}[h]
\begin{center}
\noindent \framebox{
\begin{minipage}{14cm}
\begin{dfn}[#1]
\ \\ \\
#2
\end{dfn}
\end{minipage}
}
\end{center}
\end{figure}
}

\begin{document}

\title{A new exponential separation between quantum and classical one-way communication complexity}

\author{Ashley Montanaro\footnote{Centre for Quantum Information and Foundations, Department of Applied Mathematics and Theoretical Physics, University of Cambridge, UK; {\tt am994@cam.ac.uk}. Part of this work done while at the University of Bristol.}}

\maketitle

\begin{abstract}
We present a new example of a partial boolean function whose one-way quantum communication complexity is exponentially lower than its one-way classical communication complexity. The problem is a natural generalisation of the previously studied Subgroup Membership problem: Alice receives a bit string $x$, Bob receives a permutation matrix $M$, and their task is to determine whether $Mx=x$ or $Mx$ is far from $x$. The proof uses Fourier analysis and an inequality of Kahn, Kalai and Linial.
\end{abstract}

% ------------------------------------------------------------------------------

\section{Introduction}

The framework of communication complexity \cite{yao79,kushilevitz97} has become an exceptionally successful setting in which to prove concrete lower bounds. The special case where the communication is one-way is one of the simplest models of communication complexity, and yet is one of the most important, in particular because of its applications to lower bounds on data structures and one-pass streaming algorithms \cite{kushilevitz97,baryossef04a,muthukrishnan05}. In this model, there are two parties, Alice and Bob, each of whom receives an input ($x$ and $y$, respectively). Their goal is to compute a boolean function $f(x,y)$ with success probability at least $1-\epsilon$, for some constant $\epsilon$ (usually, $\epsilon=1/3$). To do so, Alice prepares a message, which may depend on $x$ and a random string $r$ (which does not depend on $x$). She then sends Bob the message. Based on the message, his input $y$ and his own string of random bits, Bob attempts to output $f(x,y)$. The one-way communication complexity of $f$ is the minimum length of a message that Alice must send, in order that Bob can achieve success probability at least $1-\epsilon$. Note that changing $\epsilon$ to another constant strictly less than $1/2$ can only change the communication complexity by a constant factor.

In a variety of other models of communication complexity, it is known that sending {\em quantum} messages can allow Alice and Bob to drastically reduce the amount of communication required \cite{buhrman98,raz99,buhrman01}. However, in the most natural case where $f$ is a total function (i.e.\ there is no promise on Alice and Bob's inputs) it is still unknown whether quantum communication can reduce the one-way communication complexity by more than a factor of 2 \cite{winter04}. On the other hand, in the case of partial functions (where there is a promise on the inputs), it is known that the separation can be exponential \cite{kerenidis06,gavinsky06a,gavinsky08a}. Unfortunately, an exponential separation has been shown for only one or two partial functions. The first separation was shown for variants of the so-called Boolean Hidden Matching problem (see Section \ref{sec:previous} for a definition of one such variant). This problem was originally conjectured to give such a separation by Bar-Yossef, Jayram and Kerenidis \cite{baryossef04}, who indeed write (with respect to a closely related problem\footnote{A relational version of Boolean Hidden Matching, for which they do prove an exponential separation.}) ``This problem is new and we believe that its definition plays the major role in obtaining our result''.

It is therefore of great interest to find other problems which demonstrate a separation between quantum and classical one-way communication complexity. In this paper, we will be concerned with the following partial boolean function.

\boxdfn{{\sc Perm-Invariance}}{
\label{dfn:perminvar}
The {\sc Perm-Invariance} problem is defined as follows, in terms of a parameter $\beta$.
\begin{itemize}
\item Alice gets an $n$-bit string $x$.
\item Bob gets an $n \times n$ permutation matrix $M$.
\item Bob has to output $\begin{cases}1 & \text{if $Mx=x$}\\0 & \text{if $d(Mx,x) \ge \beta |x|$}\\ \text{anything} & \text{otherwise.}\end{cases}$
\end{itemize}
}

Here $|x|$ is the Hamming weight of $x$, and $d(Mx,x) = |Mx + x|$ is the Hamming distance between $Mx$ and $x$. There is a simple $O(\log n)$ qubit bounded-error one-way quantum protocol for this problem. Alice sends Bob the state $\ket{\psi_x} = \frac{1}{\sqrt{|x|}} \sum_{i,x_i=1} \ket{i}$, and Bob attaches an ancilla qubit in the state $\frac{1}{\sqrt{2}}(\ket{0}+\ket{1})$. He then applies a controlled-$M$ operation (controlled on the ancilla qubit) to produce the state $\frac{1}{\sqrt{2}}(\ket{0}\ket{\psi_x} + \ket{1}\ket{\psi_{Mx}})$. Next he performs a Hadamard operation on the ancilla, then measures it. It is easy to verify that the measurement result is 0 with probability $\frac{1}{2} + \frac{1}{2} \ip{\psi_x}{\psi_{Mx}} = 1 - d(Mx,x)/(4|x|)$, which is equal to 1 if $Mx=x$, and at most $1-\beta/4$ if $d(Mx,x) \ge \beta |x|$. For constant $\beta$, it suffices to repeat this protocol a constant number of times to determine which is the case with an arbitrarily high constant probability of success.

Our main result is the following theorem.

\begin{thm}
\label{thm:main2}
Any one-way classical protocol that solves {\sc Perm-Invariance} with $\beta = 1/8$, and a constant success probability strictly greater than $1/2$, must communicate at least $\Omega(n^{7/16})$ bits.
\end{thm}

We therefore have an exponential separation between quantum and classical one-way communication complexity for this problem. We conjecture that this lower bound is not tight, and the correct bound is $\Omega(\sqrt{n})$. This would be tight by a result of Raz \cite{raz99} that any partial function with an $O(\log n)$ qubit bounded-error one-way quantum protocol has an $O(\sqrt{n})$ bit bounded-error one-way classical protocol.

% ------------------------------------------------------------------------------

\subsection{Connection to {\sc Subgroup Membership}}

{\sc Subgroup Membership} \cite{watrous00,aaronson09a} is a problem which has been conjectured to give an asymptotic separation between quantum and classical one-way communication complexity for a total function. The problem is defined below, in terms of a group $G$.

\boxdfn{{\sc Subgroup Membership} \cite{aaronson09a}}{
The {\sc Subgroup Membership} problem is defined in terms of a group $G$, as follows.
\begin{itemize}
\item Alice gets a subgroup $H \le G$.
\item Bob gets an element $g \in G$.
\item Bob has to output 1 if $g \in H$, and 0 otherwise.
\end{itemize}
}

There is an easy one-way classical protocol for {\sc Subgroup Membership} that uses $O(\log^2|G|)$ bits of communication, and it has been conjectured \cite{aaronson09a} that for certain groups $G$ this protocol is optimal. On the other hand, there is a one-way quantum protocol that uses only $O(\log |G|)$ qubits of communication, so this would imply a quadratic separation between one-way quantum and classical communication complexity for a total function. A difficulty with proving this conjecture is that the particular structure of $G$ plays a role in the complexity of solving {\sc Subgroup Membership}, and indeed for certain groups there does exist a $O(\log |G|)$-bit classical protocol for the problem \cite{aaronson09a}.

{\sc Subgroup Membership} is a special case of {\sc Perm-Invariance}. To see this, note that if Alice gets a $|G|$ bit string $x$ which is indexed by elements of $G$, and takes the value 1 on elements $y \in H$, and Bob's permutation $M$ corresponds to the map that sends elements $z \in G$ to $gz\in G$, then $g \in H$ if and only if $Mx=x$. Otherwise, $d(Mx,x)=2|x|$. The {\sc Perm-Invariance} problem is thus a natural generalisation of {\sc Subgroup Membership}, which removes the group structure (at the expense of having to put in the promise that $d(Mx,x)$ is large ``by hand'' rather than having it guaranteed by the axioms of group theory). This lack of structure perhaps accounts for the fact that an exponential quantum-classical separation can be proven for this more general problem.

% ------------------------------------------------------------------------------

\subsection{Proof technique}
\label{sec:technique}

In order to prove Theorem \ref{thm:main2}, we restrict to the following, potentially simpler problem.

\boxdfn{{\sc PM-Invariance}}{
\label{dfn:pminvar}
The {\sc PM-Invariance} problem is defined as follows.
\begin{itemize}
\item Alice gets a $2n$-bit string $x$ such that $|x|=n$.
\item Bob gets a $2n \times 2n$ permutation matrix $M$, where the permutation entirely consists of disjoint transpositions (i.e.\ corresponds to a perfect matching on the complete graph on $2n$ vertices).
%\item Bob must output 1 if $Mx=x$, and 0 otherwise, given the promise that either $Mx=x$ or $d(Mx,x) \ge n/8$.
\item Bob has to output $\begin{cases}1 & \text{if $Mx=x$}\\0 & \text{if $d(Mx,x) \ge n/8$}\\ \text{anything} & \text{otherwise.}\end{cases}$
\end{itemize}
}

The constant $\beta = 1/8$ is essentially arbitrary here and is an artifact of the proof technique. Note that for odd $n$ the problem is trivial, as it is impossible that $Mx=x$. For the rest of the paper, we therefore assume that $n$ is even.

Theorem \ref{thm:main2} is immediate from the following result.

\begin{thm}
\label{thm:main}
Any one-way classical protocol that solves the {\sc PM-Invariance} problem and communicates at most $n^{7/16}/(e \ln 2) - O(\log n)$ bits can achieve a success probability of at most $0.79 + o(1)$.
%Any one-way classical protocol that solves {\sc Perm-Invariance} with a constant success probability strictly greater than $1/2$ must communicate at least $\Omega(n^{1/4})$ bits.
\end{thm}

The following explicit one-way classical protocol for the {\sc PM-Invariance} problem achieves a constant success probability using $O(n^{1/2})$ bits of communication: Alice simply sends Bob a randomly selected $O(n^{1/2})$ bits of her input $x$ (using shared randomness to select the bits, which by Newman's theorem can be replaced with private randomness with a negligible overhead~\cite{kushilevitz97}). By the birthday paradox, with constant probability Bob's permutation $M$ will interchange at least two of the bits that Alice sent. If $Mx=x$, then these bits will always be equal. However, if $d(Mx,x) \ge n/8$, then with constant probability the bits will not be equal. A constant number of repetitions therefore suffices to distinguish these cases with any desired constant probability.

The overall technique used to prove Theorem \ref{thm:main} is common to other works that separate quantum and classical one-way communication complexity \cite{kerenidis06,gavinsky06a,gavinsky08a}. We fix two ``hard'' distributions $\mathcal{D}_0$, $\mathcal{D}_1$ on the zero/one-valued inputs respectively, and show that for any short message from Alice to Bob, Bob does not get enough information to be able to distinguish between the induced distributions on his own inputs.

The main technical tool used is Fourier analysis, and in particular the Fourier spectrum inequality of Kahn, Kalai and Linial \cite{kahn88} (the KKL Lemma), which has found many applications in computer science, and specifically to communication complexity \cite{raz95,kerenidis06,gavinsky08a}. The relevance of Fourier analysis is as follows. A short message from Alice specifies a large subset $S$ of her inputs. Let $\mathcal{D}_1^S$ denote Bob's induced distribution on one-valued inputs, given that Alice's input is in $S$. In the case of the {\sc PM-Invariance} problem, it turns out that the distance of $\mathcal{D}_1^S$ from the uniform distribution can be upper-bounded by looking at the Fourier transform of $f$, the characteristic function of $S$.

There are several technical ingredients that need to go into this. First, we use the KKL Lemma to bound the Fourier weight of $f$ (Lemma \ref{lem:uberkkl}). The bound for the distance of $\mathcal{D}_1^S$ from the uniform distribution turns out to depend on Krawtchouk polynomials \cite{macwilliams83,krasikov99}, and the second ingredient is the use of upper bounds on these polynomials. The final ingredient is the use of a number of inequalities and identities involving perfect matchings and binomial coefficients.

% ------------------------------------------------------------------------------

\subsection{Relation to previous work}
\label{sec:previous}

Prior to this work, the only known exponential separation between one-way quantum and classical communication complexity for a (partial) boolean function was a result of Gavinsky, Kempe and de Wolf~\cite{gavinsky06a}, and independently Kerenidis and Raz \cite{kerenidis06}, on (variants of) the so-called Boolean Hidden Matching problem. These works were later combined as \cite{gavinsky08a}. The Boolean Hidden Matching problem was originally defined by Bar-Yossef, Jayram and Kerenidis \cite{baryossef04}, who also proved an exponential separation for a related relational problem called the Hidden Matching problem~\cite{baryossef04}. The variant of Boolean Hidden Matching for which a lower bound was proven in \cite{gavinsky08a} is called {\sc $\alpha$-Partial Matching}, and is given below. Gavinsky et al show in \cite{gavinsky08a} that, for any $0 < \alpha \le 1/4$, any classical one-way bounded-error protocol for the {\sc $\alpha$-Partial Matching} problem must communicate at least $\Omega(\sqrt{n/\alpha})$ bits, whereas there is a quantum protocol that transmits only $O((\log n)/\alpha)$ bits, thus proving an exponential separation for constant $\alpha$. The separation we give here for {\sc PM-Invariance} is not quite as large (and we conjecture it is not tight; see Section \ref{sec:conc}).

\boxdfn{{\sc $\alpha$-Partial Matching} \cite{gavinsky08a}}{
\label{dfn:alphapm}
The {\sc $\alpha$-Partial Matching} problem is defined as follows.
\begin{itemize}
\item Alice gets an $n$-bit string $x$.
\item Bob gets an $\alpha n\times n$ matrix $M$ over $\F_2$, where each row contains exactly two 1s, and each column contains at most one 1, and a string $w \in \{0,1\}^{\alpha n}$.
%\item They are promised that either $Mx = w$, or $Mx = w + 1^{\alpha n}$, and must determine which is the case.
\item Bob has to output $\begin{cases}0 & \text{if $Mx = w$}\\1 & \text{if $Mx = w + 1^{\alpha n}$}\\ \text{anything} & \text{otherwise.}\end{cases}$
\end{itemize}
}

The {\sc PM-Invariance} problem can be rephrased to seem quite similar to {\sc $1/2$-Partial Matching}. Rewrite Bob's permutation matrix $M$ as an $n \times 2n$ matrix $N$ such that the $i$'th row corresponds to the $i$'th pair of elements $(a_i,b_i)$ swapped by the permutation, and contains 1s in columns $a_i$ and $b_i$ (and is zero elsewhere). Then it is easy to see that $d(Mx,x) = 2|Nx|$, and the {\sc PM-Invariance} problem reduces to distinguishing between the cases $|Nx|=0$, $|Nx| \ge n/16$. So the main difference between {\sc $\alpha$-Partial Matching} and {\sc PM-Invariance} is the relaxation of the promise on Bob's input, by removing the string $w$.

This relaxation is then the main motivation for this work. First, this allows a quantum-classical separation to be proven for a direct and natural generalisation of {\sc Subgroup Membership}; it is not clear that a similar connection exists between {\sc $\alpha$-Partial Matching} and {\sc Subgroup Membership}. Second, relaxing the promise reduces the gap between partial functions (for which we have an exponential separation) and total functions (for which the best known separation is only constant). Third, given the dearth of quantum-classical communication complexity separations, it seems to be of interest to generalise and extend any known separation as far as possible.

The apparently minor change from {\sc $\alpha$-Partial Matching} to {\sc PM-Invariance} seems to increase the difficulty of proving a communication complexity lower bound. Previous proofs of the lower bound on {\sc $\alpha$-Partial Matching} work by showing that, for some initial distribution on the inputs, Bob's induced distribution on the string $Mx$ is close to uniform, and thus it is hard for him to distinguish the cases $Mx = w$ and $Mx = w + 1^{\alpha n}$. For the {\sc PM-Invariance} problem this approach does not seem to work, and it appears necessary to work directly with Bob's distribution on $M$ (rather than $Mx$). We thus obtain a fourth motivation: the development of techniques which may be of use elsewhere.

Many of the ingredients used in the beautiful proofs of the lower bound on {\sc $\alpha$-Partial Matching} given in \cite{kerenidis06,gavinsky08a} also appear in the present paper, and in particular these previous papers also make use of Fourier analysis and the KKL Lemma (indeed, the third proof in \cite{gavinsky06a} also does so implicitly, via a lemma of Talagrand). However, the {\sc PM-Invariance} problem presents us with two further technical challenges. First, it seems essential to find a bound on the Fourier weights of a boolean function that applies at both low and high levels. Second, we need to carefully bound some expressions involving Krawtchouk polynomials and binomial coefficients, in order to get a non-trivial final result.

Following the completion of this work, it was shown by Klartag and Regev \cite{klartag10} that one-way quantum communication can be exponentially stronger than even {\em two-way} classical communication. The problem they used to demonstrate this is complete for one-way quantum communication complexity.

We also note that the problem considered here is a very special variant of a partial function termed $\mathcal{P}_1$, which Raz \cite{raz99} used to give the first exponential separation between quantum and classical {\em two-way} communication complexity for a partial boolean function. In this problem, Alice gets a unit vector $x$ and two orthogonal subspaces $S_0$, $S_1$, and Bob gets an orthogonal matrix $T$. Their goal is to answer 0 if $Tx$ is within constant distance of $S_0$, and 1 if $Tx$ is within constant distance of $S_1$.

The rest of the paper is devoted to the proof of Theorem \ref{thm:main}. In Section \ref{sec:comm}, we prove a general lemma relating communication complexity to distinguishability of probability distributions. Sections \ref{sec:comb} and \ref{sec:fourier} contain the combinatorial and Fourier-analytic results we need for the proof, the heart of which finally appears in Section \ref{sec:proof}. We finish with some concluding remarks in Section \ref{sec:conc}.

% ------------------------------------------------------------------------------

\subsection{Miscellaneous notation}

We often use the notation $[E]$ for a term which evaluates to 0 if the expression $E$ is false, and 1 if $E$ is true. The $\ell_1$ distance between two vectors $p$, $q$ is defined as $\|p - q\|_1 = \sum_i |p_i - q_i|$. $[n]$ denotes the set $\{1,\dots,n\}$, and $S^c$ denotes the complement $[n]\backslash S$ of the set $S$ in $[n]$. We will continue to associate permutations of $2n$ elements which entirely consist of disjoint transpositions with perfect matchings on the complete graph on $2n$ vertices. The set of all such perfect matchings will be denoted $\text{PM}_{2n}$.

% ------------------------------------------------------------------------------

\section{Communication complexity}
\label{sec:comm}

The first step in the proof is to go from the existence of an efficient communication protocol to the existence of a large subset of Alice's inputs such that two ``hard'' input distributions are distinguishable over that subset. The following lemma achieves this in a quite general setting; this is fairly standard, and similar (but somewhat more specific) statements have been proven in previous work \cite{kerenidis06,gavinsky06a,gavinsky08a}. For any subset $S$ of Alice's inputs, and any joint distribution $\mathcal{D}$ on Alice and Bob's inputs, let $\mathcal{D}^S$ denote the distribution on Bob's inputs induced by conditioning on the event that Alice's input is in set $S$.

\begin{lem}
\label{lem:commbound}
Let $f:\{0,1\}^m \times \{0,1\}^n \rightarrow \{0,1\}$ be a function of Alice and Bob's distributed inputs. Let $\mathcal{D}_0$, $\mathcal{D}_1$ be distributions on the zero/one-valued inputs, respectively, that are each uniform over Alice's inputs, when averaged over Bob's inputs. Assume there is a one-way classical protocol that computes $f$ with success probability $1-\epsilon$, for some $\epsilon<1/3$, and uses $c$ bits of communication. Then there exists an $S \subseteq \{0,1\}^m$ such that $|S| \ge \epsilon\,2^{m-c}$, and $\|\mathcal{D}^S_0 - \mathcal{D}^S_1\|_1 \ge 2(1-3 \epsilon)$.
\end{lem}

\begin{proof}
By the Yao principle \cite{yao79}, for any distribution $\mathcal{D}$ on Alice and Bob's inputs, there is a deterministic one-way protocol $\mathcal{P}$ that communicates $c$ bits and computes $f$ correctly on a $1-\epsilon$ fraction of the inputs (with respect to $\mathcal{D}$). We choose the distribution $\mathcal{D} = \frac{1}{2}\left(\mathcal{D}_0 + \mathcal{D}_1 \right)$. Each potential message that Alice might send to Bob identifies a subset of her inputs. For each of Alice's inputs $x$, let $S_x$ denote the subset of Alice's inputs identified by the message sent on input $x$. Let $\text{Out}_S(y):\{0,1\}^n \rightarrow \{0,1\}$ be the function which takes the value that Bob outputs when he receives an input $y$ and a message from Alice that corresponds to the subset $S \subseteq \{0,1\}^m$ of her inputs. Let $p_{xy}$, $q_{xy}$ denote the probability that Alice and Bob receive inputs $(x,y)$ under distributions $\mathcal{D}_0$, $\mathcal{D}_1$, respectively. Then, as they compute $f$ correctly on at most a $1-\epsilon$ fraction of the inputs under distribution $\mathcal{D}$,
\[ \frac{1}{2} \left( \sum_{x \in \{0,1\}^m} \sum_{y \in \{0,1\}^n} p_{xy} (1 - \text{Out}_{S_x}(y)) + q_{xy} \text{Out}_{S_x}(y) \right) \ge 1-\epsilon, \]
which implies
\[ \sum_{x \in \{0,1\}^m} \sum_{y \in \{0,1\}^n} \text{Out}_{S_x}(y)(q_{xy} - p_{xy}) \ge 1 - 2\epsilon. \]
Let $\mathcal{F}$ be the family of subsets of $\{0,1\}^m$ that corresponds to the partition of Alice's inputs into subsets determined by protocol $\mathcal{P}$. As $\mathcal{P}$ communicates $c$ bits, $|\mathcal{F}| \le 2^c$. Then
\[ \sum_{S \in \mathcal{F}} \sum_{y \in \{0,1\}^n} \text{Out}_S(y)\left( \sum_{x \in S} q_{xy} - \sum_{x \in S} p_{xy}\right) \ge 1-2 \epsilon, \]
which clearly implies
\[ \sum_{S \in \mathcal{F}} \sum_{y \in \{0,1\}^n} \left| \sum_{x \in S} p_{xy} - \sum_{x \in S} q_{xy}\right| \ge 2(1-2 \epsilon). \]
Now let $p^S_y$, $q^S_y$ denote the probability that Bob receives input $y$, conditioned on Alice's input being in set $S$, under distributions $\mathcal{D}_0$, $\mathcal{D}_1$ respectively. Then
\beas p^S_y &=& \Pr_{\mathcal{D}_0}[\text{Bob's input is }y|\text{Alice's input is in }S]\\
&=& \frac{\Pr_{\mathcal{D}_0}[\text{Bob's input is }y \cap \text{Alice's input is in }S]}{\Pr_{\mathcal{D}_0}[\text{Alice's input is in }S]}\\
&=& \frac{2^m \sum_{x \in S} p_{xy}}{|S|},
\eeas
and similarly for $q_{xy}$, where we use the fact that both $\mathcal{D}_0$ and $\mathcal{D}_1$ are uniform over Alice's inputs, when averaged over Bob's inputs. Thus
\[ \sum_{S \in \mathcal{F}} \frac{|S|}{2^m} \sum_{y \in \{0,1\}^n} \left| p^S_y - q^S_y \right| \ge 2(1-2 \epsilon), \]
or more succinctly
\[ \sum_{S \in \mathcal{F}} \frac{|S|}{2^m} \|\mathcal{D}^S_0 - \mathcal{D}^S_1\|_1 \ge 2(1-2\epsilon). \]
We now split the sum up depending on whether $|S|<s$ or $|S| \ge s$, for some integer $s$ to be determined. Then
\beas
2(1-2 \epsilon) &\le& \sum_{S \in \mathcal{F},|S| < s} \frac{|S|}{2^m} \|\mathcal{D}^S_0 - \mathcal{D}^S_1\|_1 + \sum_{S \in \mathcal{F},|S| \ge s} \frac{|S|}{2^m} \|\mathcal{D}^S_0 - \mathcal{D}^S_1\|_1\\
&\le& 2 \sum_{S \in \mathcal{F},|S| < s} \frac{|S|}{2^m} + \max_{S,|S| \ge s} \|\mathcal{D}^S_0 - \mathcal{D}^S_1\|_1\\
&\le& s\,2^{c-m+1} + \max_{S,|S| \ge s} \|\mathcal{D}^S_0 - \mathcal{D}^S_1\|_1.
\eeas
Thus there exists an $S$ with $|S| \ge s$ such that
\[ \|\mathcal{D}^S_0 - \mathcal{D}^S_1\|_1 \ge 2(1-2 \epsilon - s\,2^{c-m}). \]
Taking $s = \epsilon\,2^{m-c}$, the proof is complete.
\end{proof}

% ------------------------------------------------------------------------------

\section{Combinatorial preliminaries}
\label{sec:comb}

We will need to calculate and estimate a number of combinatorial quantities to prove Theorem \ref{thm:main}. We start with some easy calculations related to perfect matchings, which we state without proof.

\begin{itemize}
\item The number of perfect matchings on the complete graph with $2n$ vertices is
\[ N_{2n} := \frac{(2n)!}{2^n n!} = (2n-1)(2n-3)\dots1. \]
\item For any $2n$-bit string $x$ with Hamming weight $2k$, and any integer $d$,
\[ \sum_{M \in \text{PM}_{2n}} [d(Mx,x) = d] = \begin{cases} N_{2n} 2^{d/2} \frac{\binom{n}{d/2} \binom{n-d/2}{k-d/4}}{\binom{2n}{2k}} & \text{if $d$ is a multiple of 4}\\0& \text{otherwise.}\end{cases}\]
\item In particular, for any $2n$-bit string $x$ with Hamming weight $2k$,
\[ \sum_{M \in \text{PM}_{2n}} [Mx=x] = N_{2k}\,N_{2(n-k)} = N_{2n} \frac{\binom{n}{k}}{\binom{2n}{2k}}. \]
\item For any pair of $2n$-bit strings $x$ and $y$, both with Hamming weight $2k$, let $t_{ab}$ be the number of bits where $x$ is equal to $a$ and $y$ is equal to $b$. Then
\[ \sum_{M \in \text{PM}_{2n}} [Mx=x,My=y] = N_{t_{00}}\,N_{t_{01}}\,N_{t_{10}}\,N_{t_{11}}. \]
This implies that
\[ \sum_{M \in \text{PM}_{2n}} [Mx=x,My=y] = \begin{cases} N_{|x \wedge y|}\,N_{2k-|x \wedge y|}^2\,N_{2n-4k+|x \wedge y|} & \text{if $|x \wedge y|$ is even}\\ 0 & \text{if $|x \wedge y|$ is odd,}\end{cases} \]
implying in turn that
\be
\label{eqn:comb} \sum_{M \in \text{PM}_{2n}} [Mx=x,My=y] = \begin{cases} N_{2(k-\ell)}\,N_{2\ell}^2 N_{2(n-k-\ell)} & \text{if $d(x,y)=4\ell$} \\ 0 & \text{otherwise.}\end{cases} \ee
\end{itemize}

We now turn to finding some technical upper bounds on quantities related to binomial coefficients.

\begin{lem}
\label{lem:binom}
For any integers $n,k \ge 0$,
\[ \frac{\binom{4n}{2n}\binom{n}{k}^2 \binom{4n}{4k} }{ \binom{2n}{n} \binom{2n}{2k}^3 } \le 2^{2n}. \]
\end{lem}

\begin{proof}
Evaluating all the binomial coeffients, the left-hand side is equal to
\[ \frac{(4n)!^2 n!^4}{(2n)!^6} \frac{(2k)!^3}{k!^2 (4k)!} \frac{2(n-k)!^3}{(n-k)!^2 (4(n-k))!}, \]
which can eventually be written out as
\[ 2^{2n} \frac{ \left(\frac{4n-1}{4n}\right)^2 \left(\frac{4n-3}{4n-2}\right)^2 \left(\frac{4n-5}{4n-4}\right)^2 \dots \left(\frac{2n+1}{2n+2}\right)^2 }{ \left(\frac{4k-1}{4k}\right)\left(\frac{4k-3}{4k-2}\right) \dots \left(\frac{2k+1}{2k+2}\right) \left(\frac{4(n-k)-1}{4(n-k)}\right) \left(\frac{4(n-k)-3}{4(n-k)-2}\right) \dots \left(\frac{2(n-k)+1}{2(n-k)+2}\right) }. \]
One can show with some tedious algebra that the denominator decreases with $k$ for $0 \le k \le n/2$, and that the overall maximum is therefore found at $k=n/2$ (for $n$ even; when $n$ is odd, this maximum is not actually achieved). Substituting this value of $k$ and simplifying, we need to show
\[ \frac{(4n-1)(4n-3)(4n-5)\dots(2n+1)(2n)(2n-2)(2n-4)\dots(n+2)}{(4n)(4n-2)(4n-4)\dots(2n+2)(2n-1)(2n-3)(2n-5)\dots(n+1)} \le 1. \]
This would follow from showing that, for any $0 \le a \le n/2-1$,
\[ \frac{(4n-4a-1)(4n-4a-3)}{(4n-4a)(4n-4a-2)} \le \frac{2n-2a-1}{2n-2a}, \]
which is equivalent to the inequality
\[ (4n-4a-1)(4n-4a-3) \le (4n-4a-2)^2. \]
It can easily be verified that this inequality holds for all $a$.
\end{proof}

\begin{lem}
\label{lem:binom3}
For any integer $n \ge 0$,
\[ \sum_{k=0}^n \frac{\binom{n}{k}^2}{\binom{2n}{2k}} = \frac{2^{2n}}{\binom{2n}{n}}. \]
\end{lem}

\begin{proof}
Using the identity
\be
\label{eqn:binom}
\binom{2n}{n} \binom{n}{k}^2 = \binom{2k}{k}\binom{2n-2k}{n-k}\binom{2n}{2k},
\ee
the lemma reduces to the statement that
\[ \sum_{k=0}^n \binom{2k}{k} \binom{2n-2k}{n-k} = 2^{2n}, \]
which is equation (5.39) in \cite{graham04}. But identity (\ref{eqn:binom}) follows from applying a succession of ``trinomial revision'' identities \cite{graham04}:
\[ \binom{2n}{n} \binom{n}{k}^2 = \binom{2n}{k}\binom{2n-k}{n}\binom{n}{k} = \binom{2n}{k}\binom{2n-k}{k}\binom{2n-2k}{n-k} = \binom{2n}{2k}\binom{2k}{k}\binom{2n-2k}{n-k}. \]
There is an alternative combinatorial proof of this identity, which we leave to the interested reader.
\end{proof}

The following inequality is well-known, but we include a proof for completeness.

\begin{lem}
\label{lem:binom2}
For any integers $n,k \ge 0$,
\[ \binom{n}{k}^2 \le \binom{2n}{2k}. \]
\end{lem}

\begin{proof}
There is a simple combinatorial proof of this statement\footnote{Thanks to Ronald de Wolf for pointing this out.}. The right-hand side counts the number of ways of choosing $2k$ elements from a set $S$ of size $2n$, while the left-hand side counts the number of ways of choosing $k$ elements from the first $n$ elements of $S$, and $k$ elements from the last $n$ elements of $S$. The latter is clearly upper bounded by the former.
\end{proof}

Finally, we will need to evaluate some sums involving binomial coefficients.

\begin{lem}
\label{lem:genfn}
For any $n \ge 0$,
\[ \sum_k k \binom{n}{2k} = \begin{cases}\frac{1}{8}\left(n2^n \right)& \text{if }n \ge 2\\0& \text{if }n\le1, \end{cases} \]
and also
\[ \sum_k k^2 \binom{n}{2k} = \begin{cases}\frac{1}{32}\left(n(n+1)2^n \right)& \text{if }n \ge 3\\1& \text{if }n=2\\0& \text{if }n\le1. \end{cases} \]
\end{lem}

\begin{proof}
We prove this lemma using the method of generating functions \cite{graham04}. Consider the function
\[ f(x) := nx(1+x)^{n-1} = x \frac{d}{dx} (1+x)^n = \sum_k k \binom{n}{k} x^k. \]
The quantity $\sum_k (2k) \binom{n}{2k} x^{2k}$ is equal to the sum of the terms of this series that correspond to even $k$. The function $\frac{1}{2}\left(f(x) + f(-x)\right)$ extracts precisely these terms. Thus
\[ \sum_k k \binom{n}{2k} x^{2k} = \frac{1}{4}\left(f(x) + f(-x)\right) = \frac{nx}{4} \left( (1+x)^{n-1} - (1-x)^{n-1} \right). \]
Substituting $x=1$ proves the first part of the lemma. For the second part, differentiate $f$ again, and multiply by $x$, to obtain
\[ g(x) := n x(1+nx)(1+x)^{n-2} = x \frac{d}{dx} nx(1+x)^{n-1} = \sum_k k^2 \binom{n}{k} x^k. \]
As before, considering $\frac{1}{2}\left(g(x) + g(-x)\right)$ we get
\[ \sum_k k^2 \binom{n}{2k} x^{2k} = \frac{1}{8}\left(g(x) + g(-x)\right) = \frac{nx}{8}\left((1+nx)(1+x)^{n-2} - (1-nx)(1-x)^{n-2} \right). \]
Substituting $x=1$ proves the second part of the lemma.
\end{proof}

% ------------------------------------------------------------------------------

\section{Fourier analysis}
\label{sec:fourier}

For a function $f:\{0,1\}^n \rightarrow \R$, we define the Fourier transform of $f$ by
\[ \hat{f}(S) = \frac{1}{2^n} \sum_{x \in \{0,1\}^n} (-1)^{\sum_{i \in S} x_i} f(x), \]
for $S \subseteq [n]$. Subsets of $[n]$ are in obvious correspondence with $n$-bit strings, and we sometimes use the notation $\hat{f}(s)$ to imply the identification of $S$ with its characteristic vector $s$. For any functions $f,g:\{0,1\}^n \rightarrow \R$, it is easy to show that
\be
\label{eq:fourier} \sum_{x,y \in \{0,1\}^n} f(x) f(y) g(x + y) = 2^{2n} \sum_{S \subseteq [n]} \hat{g}(S) \hat{f}(S)^2. 
\ee
Let $W_k(f)$ be the Fourier weight of $f$ at level $k$, i.e.\ $W_k(f) = \sum_{S, |S|=k} \hat{f}(S)^2$. 

% ------------------------------------------------------------------------------

\subsection{Krawtchouk polynomials}
\label{sec:krawtchouk}

We will use properties of the Krawtchouk polynomials to obtain our bounds. The $k$'th Krawtchouk polynomial $K_k^n$ is defined as the unique degree $k$ polynomial satisfying
\[ K_k^n(x) = \sum_{i=0}^k (-1)^i \binom{x}{i} \binom{n-x}{k-i} \]
for integer $x$, and the Krawtchouk transform of a function $g:\{0,1,\dots,n\} \rightarrow \R$ is given by the function $h:\{0,1,\dots,n\} \rightarrow \R$ defined by
\[ h(x) = \frac{1}{2^n} \sum_{k=0}^n K_k^n(x) g(k). \]
Krawtchouk polynomials are important for us because, if $f:\{0,1\}^n \rightarrow \R$ is a symmetric function, i.e.\ $f(x) = g(|x|)$ for some $g$, then the Fourier transform of $f$ is given by the Krawtchouk transform of $g$:
\[ \hat{f}(s) = \frac{1}{2^n} \sum_{x \in \{0,1\}^n} (-1)^{x \cdot s} f(x) = \frac{1}{2^n} \sum_{k=0}^n \left( \sum_{x \in \{0,1\}^n,|x|=k} (-1)^{x \cdot s} \right) g(k) = \frac{1}{2^n} \sum_{k=0}^n K_k^n(|s|) g(k). \]
The Krawtchouk polynomials satisfy many identities and inequalities \cite{macwilliams83,krasikov99}. In particular, it holds that $K_k^n(x) = (-1)^x K_{n-k}^n(x)$, and we have the orthogonality relation
\[ \sum_{x=0}^n \binom{n}{x} K_r^n(x) K_s^n(x) = 2^n \binom{n}{r} \delta_{rs}. \]
%
%the following inequality is immediate \cite{krasikov99}:
%
%\be
%\label{eq:krawtineq}
%K_k^n(x)^2 \le 2^n \frac{\binom{n}{k}}{\binom{n}{x}}.
%\ee
%
We will need the following explicit expressions for some of the Krawtchouk polynomials:
\[ K_0^n(x) = 1,\;\; K_2^n(x) = \binom{n}{2} - 2nx + 2x^2,\;\;K_n^{2n}(x) = \begin{cases} (-1)^{x/2} \binom{2n}{n}\frac{\binom{n}{x/2}}{\binom{2n}{x}} & \text{if $x$ is even}  \\ 0 & \text{if $x$ is odd.} \end{cases} \]
For a derivation of the last expression, see \cite{linial02}. Finally, from the symmetry relation
\[ \binom{n}{x} K_k^n(x) = \binom{n}{k} K_x^n(k) \]
one can deduce
\[ K_k^n(0) = \binom{n}{k},\;\; K_k^n(2) = \frac{1}{n(n-1)} \binom{n}{k}\left( (n-2k)^2 - n \right).\]
%

% ------------------------------------------------------------------------------

\subsection{Upper bounds on Fourier weight}

We will use the following lemma of Kahn, Kalai and Linial \cite{kahn88}, which follows from the Bonami-Beckner hypercontractive inequality \cite{bonami70,beckner75}.

\begin{lem}[KKL Lemma \cite{kahn88}]
\label{lem:kkl}
Let $f:\{0,1\}^n \rightarrow \{-1,0,1\}$ be a function that takes a nonzero value at $p$ positions. Then, for any $0 \le \delta \le 1$,
\[ \sum_{S \subseteq [n]} \delta^{|S|} \hat{f}(S)^2 \le \left(\frac{p}{2^n}\right)^{\frac{2}{1+\delta}}. \]
\end{lem}

If $f$ is the characteristic function of a set, the KKL Lemma can be used to obtain quite tight bounds on the Fourier weight of $f$ at both low {\em and high} levels, which we formalise as Lemma \ref{lem:uberkkl} below. Part (i) of this lemma is well-known (e.g.\ see \cite{gavinsky08a,dewolf08}); however, part (ii) appears to be new (albeit not difficult).

\begin{lem}
\label{lem:uberkkl}
Consider an arbitrary non-empty subset $A \subseteq \{0,1\}^n$, let $f$ be the characteristic function of $A$, and set $2^{-\alpha} = \frac{|A|}{2^n}$. Then, for any $1 \le k \le (\ln 2) \alpha$:
\beas
&\text{(i)}& W_k(f) \le 2^{-2 \alpha} \left(\frac{(2e \ln 2)\alpha}{k}\right)^k, \text{ and}\\
&\text{(ii)}& W_{n-k}(f) \le 2^{-2 \alpha} \left(\frac{(2e \ln 2)\alpha}{k}\right)^k.
\eeas
\end{lem}

\begin{proof}
By the KKL Lemma, we have
\[ W_k(f) \le \delta^{-k} 2^{-2\alpha/(1+\delta)} \]
for any $k>0$ and any $0 < \delta \le 1$. We now take $\delta = \frac{\gamma}{\alpha-\gamma}$ for some $0 < \gamma \le \alpha/2$, implying
\[ W_k(f) \le \left(\frac{\alpha}{\gamma} - 1\right)^k 2^{-2(\alpha-\gamma)} < \left(\frac{\alpha}{\gamma}\right)^k 2^{-2(\alpha-\gamma)}. \]
Minimising over $\gamma$, we find that the minimum is achieved when $\gamma = \frac{k}{2 \ln 2}$, giving an upper bound
\[ W_k(f) \le 2^{-2 \alpha} \left(\frac{(2e \ln 2)\alpha}{k}\right)^k \approx 2^{-2 \alpha} \left(\frac{3.77\,\alpha}{k}\right)^k. \]
For the second part of the lemma, consider the function $g(x) = (-1)^{|x|} f(x)$. By the first part, we have $W_k(g) \le 2^{-2 \alpha} \left(\frac{(2e \ln 2)\alpha}{k}\right)^k$. We also have
\[ \hat{g}(S) = \frac{1}{2^n} \sum_{x \in \{0,1\}^n} (-1)^{\sum_{i \in S} x_i} (-1)^{\sum_{j=1}^n x_j} f(x) = \frac{1}{2^n} \sum_{x \in \{0,1\}^n} (-1)^{\sum_{i \in S^c} x_i} f(x) = \hat{f}(S^c), \]
so $W_k(g) = W_{n-k}(f)$. The second part of the lemma follows.
\end{proof}

These bounds are almost tight. Consider the $d$-dimensional subspace $S \subseteq \{0,1\}^n$ that consists of all bit strings that begin with $n-d$ zeroes. Then the Fourier transform of $S$ is uniform on the orthogonal subspace $S^\perp$, which consists of all bit strings that end with $d$ zeroes. It therefore holds that
\[ W_k(S) = \frac{1}{2^{2(n-d)}} \binom{n-d}{k} = 2^{-2 \alpha} \binom{\alpha}{k} \ge 2^{-2 \alpha} \left(\frac{\alpha}{k}\right)^k, \]
where we use a standard bound on binomial coefficients, and as before define $2^{-\alpha} = \frac{|S|}{2^n}$.
%The bounds above can thus be seen as quantifying the intuitive observation that the Fourier weight on the $k$'th level of a boolean function should not deviate from ``binomiality'' too much.

% ------------------------------------------------------------------------------

\section{Proof of Theorem \ref{thm:main}}
\label{sec:proof}

Using Lemma \ref{lem:commbound}, we will put a lower bound on the classical one-way communication complexity of the {\sc PM-Invariance} problem. The two distributions we will consider are defined as follows.

\begin{itemize}
\item $\mathcal{D}_0$: $x$ is picked uniformly at random consistent with $|x|=n$, and $M$ is a perfect matching consistent with $d(Mx,x) \ge n/8$, but otherwise uniformly random.

\item $\mathcal{D}_1$: $x$ is picked uniformly at random consistent with $|x|=n$, and $M$ is a perfect matching consistent with $Mx=x$, but otherwise uniformly random.
\end{itemize}

Let $A$ be an arbitrary subset of $\{0,1\}^{2n}$. We will show that $\|\mathcal{D}_0^A - \mathcal{D}_1^A \|_1$ cannot be large unless $A$ is small. To do this, we will show that both distributions are in fact close to uniform.

In the case of $\mathcal{D}_0$, this is quite straightforward. Let $U$ denote the uniform distribution on perfect matchings on the complete graph with $2n$ vertices. Then, for any $x \in \{0,1\}^{2n}$ such that $|x|=n$,
\[ \|\mathcal{D}_0^{\{x\}} - U\|_1 = 2 \sum_{M \in \text{PM}_{2n}} \frac{[d(Mx,x)<n/8]}{N_{2n}} = \frac{2}{\binom{2n}{n}} \sum_{d < n/32} 2^{2d} \binom{n}{2d} \binom{n-2d}{n/2-d}; \]
the first equality here is just the fact that for any probability distributions $p$, $q$, $\sum_i |p_i - q_i| = 2 \sum_{i,p_i < q_i} (q_i-p_i)$. In the second equality, for each $d < n/32$, the summand counts the number of $M \in \text{PM}_{2n}$ such that $d(Mx,x)=4d$ (as discussed in Section \ref{sec:comb}, $d(Mx,x)$ is a multiple of 4 for all $M$). This sum can be upper bounded by estimating
\[ \|\mathcal{D}_0^{\{x\}} - U\|_1 \le \frac{2}{\binom{2n}{n}} 2^{n/16} \binom{n}{n/2} \sum_{d<n/32} \binom{n}{2d} \le \frac{2}{\binom{n}{n/2}} 2^{n/16} \sum_{d<n/32} \binom{n}{2d} \le 2\sqrt{n}\,2^{n(1/16+H(1/16)-1)}, \]
where we use Lemma \ref{lem:binom2} in the penultimate inequality; in the last inequality, $H$ is the binary entropy function $H(x) = -x \log_2 x - (1-x) \log_2(1-x)$, and we use the easy inequality $\binom{n}{n/2} \ge 2^n/\sqrt{n}$. It can be verified that this quantity is exponentially small in $n$. Thus, for any $A$, we have
\[ \|\mathcal{D}_0^A - U\|_1 \le \frac{1}{|A|} \sum_{x \in A} \|\mathcal{D}_0^{\{x\}} - U\|_1 = 2^{-\Omega(n)}. \]
The case of $\mathcal{D}_1$ is more challenging. Let $p_M$ denote the probability under $\mathcal{D}_1$ that Bob gets input $M$, given that Alice's bit string was in $A$. Then
\beas
p_M &=& \Pr_{\mathcal{D}_1}[\text{Bob gets $M\,|$ Alice got something in $A$}]\\
&=& \frac{1}{|A|}\sum_{x \in A} \frac{[Mx=x]}{|\{N \in \text{PM}_{2n}:Nx=x\}|}
= \frac{\binom{2n}{n}}{N_{2n} \binom{n}{n/2} |A|}\sum_{x \in A} [Mx=x]\\
&=& \frac{\binom{2n}{n}}{N_{2n} \binom{n}{n/2}} \Pr_{x \in A} [Mx=x].
\eeas
We now attempt to upper bound $\|\mathcal{D}_1^A - U\|_1$ by appealing to the $\ell_2$ norm, using the simple inequality $\sum_{i=1}^n |x_i| \le \sqrt{n} \sqrt{ \sum_{i=1}^n x_i^2 }$ :
\[ \|\mathcal{D}_1^A - U\|_1 = \sum_{M \in \text{PM}_{2n}} \left| p_M - \frac{1}{N_{2n}} \right| \le \sqrt{N_{2n}} \sqrt{\sum_{M \in \text{PM}_{2n}} \left(p_M - \frac{1}{N_{2n}}\right)^2} = \sqrt{N_{2n} \sum_{M \in \text{PM}_{2n}} p_M^2 - 1}. \]
The interesting quantity under the square root is
\beas
N_{2n} \sum_{M \in \text{PM}_{2n}} p_M^2 &=& \frac{\binom{2n}{n}^2}{N_{2n} \binom{n}{n/2}^2} \sum_{M \in \text{PM}_{2n}} \Pr_{x \in A} [Mx=x]^2 \\
&=& \frac{\binom{2n}{n}^2}{N_{2n} \binom{n}{n/2}^2 |A|^2}\sum_{M \in \text{PM}_{2n}} \sum_{\substack{x,y \in \{0,1\}^{2n},\\|x|=|y|=n}} f(x)f(y)[Mx=x][My=y] \\
&=& \frac{\binom{2n}{n}^2}{N_{2n} \binom{n}{n/2}^2 |A|^2}\left(\sum_{x,y} f(x)f(y) \sum_{M \in \text{PM}_{2n}} [Mx=x,My=y] \right),
\eeas
where we define $f:\{0,1\}^{2n} \rightarrow \{0,1\}$ to be the characteristic function of $A$. Recalling that $\sum_{M \in \text{PM}_{2n}} [Mx=x,My=y]$ depends only on $d(x,y)$, we want to upper bound a quantity of the form
\[ \sum_{x,y} f(x) f(y) g(d(x,y)), \]
where, for $z$ a multiple of 4, $0 \le z \le 2n$, by eqn.\ (\ref{eqn:comb})
\[ g(z) = N_{z/2}^2\,N_{n-z/2}^2 = \left( \frac{(z/2)!}{2^{z/4} (z/4)!} \frac{(n-z/2)!}{2^{n/2-z/4} (n/2-z/4)!} \right)^2 = \frac{1}{2^n} \left(\frac{n!}{(n/2)!} \frac{\binom{n/2}{z/4}}{\binom{n}{z/2}} \right)^2, \]
and $g(z) = 0$ elsewhere. To find such a bound, it is convenient to use the Fourier expansion. Defining $h:\{0,1\}^{2n} \rightarrow \R$ by $h(x) = g(|x|)$, by eqn.\ (\ref{eq:fourier}) upper bounding $N_{2n} \sum_{M \in \text{PM}_{2n}} p_M^2$ is equivalent to proving an upper bound on
\[ \frac{\binom{2n}{n}^2}{N_{2n} \binom{n}{n/2}^2 |A|^2} 2^{4n} \sum_{S \subseteq [2n]} \widehat{h}(S) \hat{f}(S)^2. \]

We calculate
\[ \widehat{h}(S) = \frac{1}{2^{2n}} \sum_{k=0}^{2n} K_k^{2n}(|S|) g(k) = \frac{(n!)^2}{(n/2)!^2\,2^{3n}} \sum_{k=0}^{n/2} K_{4k}^{2n}(|S|) \frac{\binom{n/2}{k}^2}{\binom{n}{2k}^2} = \frac{n!}{2^{3n}} \binom{n}{n/2} \sum_{k=0}^{n/2} K_{4k}^{2n}(|S|) \frac{\binom{n/2}{k}^2}{\binom{n}{2k}^2}. \]
This implies that the quantity we would like to upper bound is
\[ N_{2n} \sum_{M \in \text{PM}_{2n}} p_M^2 = \frac{\binom{2n}{n} 2^{2n}}{\binom{n}{n/2} |A|^2} \sum_{S \subseteq [2n]} \sum_{k=0}^{n/2} K_{4k}^{2n}(|S|) \frac{\binom{n/2}{k}^2}{\binom{n}{2k}^2} \hat{f}(S)^2.\]
The first thing to note about this sum is that the terms with $|S|$ odd don't contribute anything; in fact, for odd $x$ it holds that
\[ \sum_{k=0}^{n/2} K_{4k}^{2n}(x) \frac{\binom{n/2}{k}^2}{\binom{n}{2k}^2} = 0. \]
To see this, recall that $K_{4k}^{2n}(x) = (-1)^x K_{2n-4k}^{2n}(x)$, which means that all the terms in this sum cancel out, except when $n$ is a multiple of 4 and $k=n/4$. But $K_n^{2n}(x) = 0$ when $x$ is odd. We are therefore left with the quantity
\be
\label{eq:alphabound}
2^{2(\alpha - n)} \frac{\binom{2n}{n}}{\binom{n}{n/2}} \sum_{s=0}^n \sum_{k=0}^{n/2} K_{4k}^{2n}(2s) \frac{\binom{n/2}{k}^2}{\binom{n}{2k}^2} W_{2s}(f),
\ee
where we set $2^{\alpha} = \frac{2^{2n}}{|A|}$. We first turn to finding an upper bound on the inner sum
\[ \sum_{k=0}^{n/2} K_{4k}^{2n}(2s) \frac{\binom{n/2}{k}^2}{\binom{n}{2k}^2}. \]
We rewrite this using the symmetry relation for Krawtchouk polynomials to get
\[ \sum_{k=0}^{n/2} K_{4k}^{2n}(2s) \frac{\binom{n/2}{k}^2}{\binom{n}{2k}^2} = \frac{1}{\binom{2n}{2s}} \sum_{k=0}^{n/2} K_{2s}^{2n}(4k) \binom{2n}{4k} \frac{\binom{n/2}{k}^2}{\binom{n}{2k}^2} \]
and apply Cauchy-Schwarz to give
\bea
\nonumber \sum_{k=0}^{n/2} K_{4k}^{2n}(2s) \frac{\binom{n/2}{k}^2}{\binom{n}{2k}^2} &\le&  \frac{1}{\binom{2n}{2s}} \left(\sum_{k=0}^{n/2} K_{2s}^{2n}(4k)^2 \binom{2n}{4k} \right)^{1/2} \left(\sum_{k=0}^{n/2} \frac{\binom{2n}{4k} \binom{n/2}{k}^4}{\binom{n}{2k}^4} \right)^{1/2}\\
\nonumber &=& \frac{2^n}{\binom{2n}{2s}^{1/2}} \left(\sum_{k=0}^{n/2} \frac{\binom{2n}{4k} \binom{n/2}{k}^4}{\binom{n}{2k}^4} \right)^{1/2} \le \frac{2^{3n/2}\binom{n}{n/2}^{1/2}}{\binom{2n}{2s}^{1/2}\binom{2n}{n}^{1/2}} \left(\sum_{k=0}^{n/2} \frac{\binom{n/2}{k}^2}{\binom{n}{2k}} \right)^{1/2}\\
\label{eq:kbound} &=& \frac{2^{2n}}{\binom{2n}{2s}^{1/2}\binom{2n}{n}^{1/2}},
\eea
where the first equality is the orthogonality relation for Krawtchouk polynomials, the second inequality is Lemma \ref{lem:binom}, and the second equality is Lemma \ref{lem:binom3}.

We will use this inequality to bound the overall sum (\ref{eq:alphabound}); however, in order to obtain a stronger upper bound, we start by treating the cases $s=0$ and $s=1$ separately. For $s=0$, we have
\[ 2^{2(\alpha - n)} \frac{\binom{2n}{n}}{\binom{n}{n/2}} \sum_{k=0}^{n/2} K_{4k}^{2n}(0) \frac{\binom{n/2}{k}^2}{\binom{n}{2k}^2} W_0(f) = 2^{2(\alpha - n)} \frac{\binom{2n}{n}}{\binom{n}{n/2}} \sum_{k=0}^{n/2} \binom{2n}{4k} \frac{\binom{n/2}{k}^2}{\binom{n}{2k}^2} 2^{-2 \alpha} \le \frac{1}{2^n} \sum_{k=0}^{n/2} \binom{n}{2k} = \frac{1}{2}, \]
where the inequality is Lemma \ref{lem:binom}. In the case $s=1$,
\[ 2^{2(\alpha - n)} \frac{\binom{2n}{n}}{\binom{n}{n/2}} \sum_{k=0}^{n/2} K_{4k}^{2n}(2) \frac{\binom{n/2}{k}^2}{\binom{n}{2k}^2} W_2(f) \le  \frac{((e \ln 2) \alpha)^2 \binom{2n}{n}}{2n(2n-1) 2^{2n} \binom{n}{n/2}} \sum_{k=0}^{n/2} \binom{2n}{4k} \left( (2n-8k)^2 - 2n \right) \frac{\binom{n/2}{k}^2}{\binom{n}{2k}^2}, \]
by Lemma \ref{lem:uberkkl} and the expression for the Krawtchouk polynomial worked out in Section \ref{sec:krawtchouk}. Dropping the negative term and using Lemma \ref{lem:binom}, we obtain an upper bound of
\[ \frac{2 ((e \ln 2) \alpha)^2}{n(2n-1) 2^n} \sum_{k=0}^{n/2} (n-4k)^2 \binom{n}{2k} = \frac{2 ((e \ln 2) \alpha)^2}{n(2n-1)} \left(\frac{n^2}{2} - n^2 + \frac{n(n+1)}{2} \right), \]
where we use Lemma \ref{lem:genfn} to evaluate the sum. We therefore have
\[ 2^{2(\alpha - n)} \frac{\binom{2n}{n}}{\binom{n}{n/2}} \sum_{k=0}^{n/2} K_{4k}^{2n}(2) \frac{\binom{n/2}{k}^2}{\binom{n}{2k}^2} W_2(f) \le \frac{((e \ln 2) \alpha)^2}{2n-1}. \]
We now bound the rest of the sum (\ref{eq:alphabound}). First, note that these inequalities for $s=0$, $s=1$ also apply to $s=n$, $s=n-1$ by symmetry considerations (the Krawtchouk polynomials are symmetric about $s=n/2$, as is Lemma \ref{lem:uberkkl}). We now apply inequality (\ref{eq:kbound}) to all of the other terms in the sum (\ref{eq:alphabound}) to obtain an upper bound of
\[ 1 + \frac{2((e \ln 2) \alpha)^2}{2n-1} + \frac{\binom{2n}{n}^{1/2}}{\binom{n}{n/2}} 2^{2\alpha} \sum_{s=2}^{n-2} \frac{W_{2s}(f)}{\binom{2n}{2s}^{1/2}}. \]
We are now ready to apply Lemma \ref{lem:uberkkl} to the terms in this sum such that either $s$ or $n-s$ is upper bounded by $\frac{1}{2}(\ln 2) \alpha$. This gives an upper bound on expression (\ref{eq:alphabound}) of
\[ 1 + \frac{2((e \ln 2) \alpha)^2}{2n-1} + \frac{\binom{2n}{n}^{1/2}}{\binom{n}{n/2}}\left( 2 \sum_{s=2}^{\lfloor \frac{1}{2}(\ln 2) \alpha \rfloor} \frac{1}{\binom{2n}{2s}^{1/2}} \left(\frac{(2e \ln 2)\alpha}{2s}\right)^{2s} + 2^{2 \alpha} \sum_{s=\lceil \frac{1}{2}(\ln 2) \alpha \rceil}^{n-\lceil \frac{1}{2}(\ln 2) \alpha \rceil} \frac{W_{2s}(f)}{\binom{2n}{2s}^{1/2}} \right), \]
where the factor of 2 in the first sum follows from using the symmetry about $n/2$ of Lemma \ref{lem:uberkkl}. We find a bound for the sum over small values of $s$ by simply lower bounding the binomial coefficients via the inequality $\binom{n}{s} \ge \left(\frac{n}{s}\right)^s$, which is valid for all $s \ge 1$:
\[ \sum_{s=2}^{\lfloor \frac{1}{2}(\ln 2) \alpha \rfloor} \frac{1}{\binom{2n}{2s}^{1/2}} \left(\frac{(2e \ln 2)\alpha}{2s}\right)^{2s} \le \sum_{s=2}^\infty \left(\frac{s}{n} \right)^s \left(\frac{(e \ln 2)\alpha}{s}\right)^{2s} = \sum_{s=2}^\infty \left(\frac{(e \ln 2)\alpha}{\sqrt{sn}}\right)^{2s}. \]
Now, for any $\alpha \le \frac{n^{7/16}}{e \ln 2}$, we have
\[ \sum_{s=2}^\infty \left(\frac{(e \ln 2)\alpha}{\sqrt{sn}}\right)^{2s} \le \frac{1}{4} \sum_{s=2}^\infty n^{-s/8} = \frac{1}{4 n^{1/8}(n^{1/8}-1)}. \]
Finally, the sum over large values of $s$ can be trivially upper bounded by noting that $\sum_{s=0}^{n} W_{2s}(f) \le 1$, and hence
\[ \sum_{s=\lceil \frac{1}{2}(\ln 2) \alpha \rceil}^{n-\lceil \frac{1}{2}(\ln 2) \alpha \rceil} \frac{W_{2s}(f)}{\binom{2n}{2s}^{1/2}} \le \frac{1}{\binom{2n}{\lfloor (\ln 2) \alpha \rfloor}^{1/2}} \le \left( \frac{2n}{\lfloor (\ln 2) \alpha \rfloor} \right)^{-\lfloor (\ln 2) \alpha \rfloor/2} = n^{-\Omega(n^{7/16})}; \]
note that here we assume that $\alpha = \Omega(n^{7/16})$ without loss of generality. By Stirling's approximation, $\frac{\binom{2n}{n}^{1/2}}{\binom{n}{n/2}} = \frac{\pi^{1/4}}{\sqrt{2}}n^{1/4} + o(1)$. The overall bound on $N_{2n} \sum_{M \in \text{PM}_{2n}} p_M^2$ thus becomes, for any $\alpha \le \frac{n^{7/16}}{e \ln 2}$,
\[ N_{2n} \sum_{M \in \text{PM}_{2n}} p_M^2 \le 1 + O(n^{-1/8}) + \left(\frac{\pi^{1/4}}{\sqrt{2}} n^{1/4} + o(1) \right) \left( \frac{1}{2 n^{1/8}(n^{1/8}-1)} + n^{-\Omega(n^{7/16})}\right), \]
implying
\[ \|\mathcal{D}_1^A - U \|_1 \le \left(N_{2n} \sum_{M \in \text{PM}_{2n}} p_M^2-1\right)^{1/2} \le \left( \frac{\pi^{1/4}}{2\sqrt{2}} + o(1) \right)^{1/2} \le  0.69 + o(1). \]
Going back to the original inequality we were trying to prove, we have shown that, for all $A$ such that
$|A| \ge 2^{2n-\frac{n^{7/16}}{e \ln 2}}$,
it holds that
\[ \|\mathcal{D}_0^A - \mathcal{D}_1^A \|_1 \le \|\mathcal{D}_0^A - U \|_1 + \|\mathcal{D}_1^A - U \|_1 \le 2^{-\Omega(n)} + \left( \frac{\pi^{1/4}}{2\sqrt{2}} + o(1) \right)^{1/2} \le 0.69 + o(1). \]
Alice's input is of length $m = \log_2 \binom{2n}{n} = 2n - O(\log n)$ bits. Thus, taking $c = \frac{n^{7/16}}{e \ln 2} - O(\log n)$ and $\epsilon = 0.21 - o(1)$ in Lemma \ref{lem:commbound}, this implies that any classical one-way protocol that transmits at most $\frac{n^{7/16}}{e \ln 2} - O(\log n)$ bits from Alice to Bob cannot compute {\sc PM-Invariance} with success probability greater than $0.79 + o(1)$. This completes the proof of Theorem \ref{thm:main}.

% ------------------------------------------------------------------------------

\section{Conclusions}
\label{sec:conc}

We have given an example of a natural problem for which there is an exponential separation between one-way quantum and classical communication complexity. However, the problem of determining whether such a separation -- or indeed {\em any} asymptotic separation -- can exist for a total function still remains.

We conjecture that the lower bound of $\Omega(n^{7/16})$ for the classical communication complexity of {\sc PM-Invariance} is not tight, and the true lower bound is $\Omega(n^{1/2})$, matching the upper bound. It appears that it would suffice to obtain stronger upper bounds on the Krawtchouk polynomials to prove such a result.

Finally, we mention an intriguing connection to coding theory. Alice's subset $A \subseteq \{0,1\}^{2n}$ can be thought of as a code\footnote{In the sense that any set of bit strings corresponds to a code; there is no constraint on the minimum distance of this code.}, in which case the quantity we upper bound is an exponentially decreasing function of the distance distribution of that code. There are a number of works that put constraints on distance distributions, and in particular show that for ``large enough'' codes, the distance distribution must be similar to that of a random code (e.g.\ see \cite{krasikov97,ashikhmin05}). Tight enough results of this form might suffice to prove our main result (as it is easy to verify that it holds for a random subset $A$). Indeed, in \cite{linial02}, Linial and Samorodnitsky show that for large {\em linear} codes, the number of bit strings at distance $n$ is (asymptotically) maximised by random codes. Note that this result does not depend on any information about the code's minimum distance. Extending the result of \cite{linial02} in a suitable way to non-linear codes (i.e.\ general subsets of $\{0,1\}^{2n}$) and other distances might allow a tight lower bound on the communication complexity of {\sc PM-Invariance} to be proven. Conversely, our result might have implications for the understanding of distance distributions of general codes.

% ------------------------------------------------------------------------------

\section*{Acknowledgements}

This work was supported by an EPSRC Postdoctoral Research Fellowship. I would like to thank Aram Harrow and Rapha\"el Clifford for helpful discussions and suggestions, and Ronald de Wolf and Oded Regev for helpful comments on a previous version. I would also like to thank two anonymous referees for comments which improved the paper.

% ------------------------------------------------------------------------------

%\bibliographystyle{plain}
%\bibliography{../thesis}

% ------------------------------------------------------------------------------
\end{document}